\documentclass[10pt]{IEEEtran}
\usepackage{amsmath,amsfonts, amsthm, amssymb}
\usepackage{algorithmic}
\usepackage{algorithm}
\usepackage{array}
\usepackage[caption=false,font=normalsize,labelfont=sf,textfont=sf]{subfig}
\usepackage{textcomp}
\usepackage{stfloats}
\usepackage{url}
\usepackage{verbatim}
\usepackage{graphicx}
\usepackage{cite}
\usepackage{caption}
\usepackage{mathtools}

\usepackage[columnsep=0.25in, top=0.75in, bottom=1.1in, textwidth=7.25in ]{geometry}
\usepackage{stmaryrd} 

\usepackage{xcolor}
\usepackage{enumitem}
\usepackage{thmtools}
\usepackage{centernot}
\usepackage{nopageno}
\declaretheoremstyle[
  headfont=\bfseries, 
  bodyfont=\normalfont,
]{boldstyle}

\declaretheorem[style=boldstyle, name=Theorem]{theorem}

\declaretheorem[style=boldstyle, name=Lemma]{lemma}
\declaretheorem[style=boldstyle, name=Remark]{remark}

\declaretheorem[style=boldstyle, name=Definition]{definition}

\newcommand{\C}{\mathcal{C}} 
\newcommand{\Entity}{\mathbf{E}} 
\newcommand{\Predicate}{\mathbf{P}} 
\newcommand{\Variable}{\mathbf{X}} 
\newcommand{\Constant}{\mathbf{C}} 
\newcommand{\Term}{\mathbf{T}} 
\newcommand{\AttConst}{\mathbf{C}_\mathbf{A}} 
\newcommand{\Model}{\mathcal{M}} 

\newcommand{\Evidence}{\mathcal{E}} 

\newcommand{\QSent}{\mathbf{Q}} 
\newcommand{\Width}{w} 
\newcommand{\Language}{\mathcal{L}} 
\newcommand{\ProbMeas}{\mathcal{P}_{I}} 
\newcommand{\ProbLoc}[1]{\mathcal{P}_{#1}} 

\newcommand{\argmax}{\mathop{\mathrm{arg\,max}}} 

\newcommand{\NumAttConst}{K}

\begin{document}

\title{Analysis of Semantic Communication for Logic-based Hypothesis Deduction

\thanks{This work is supported by National Science Foundation under award ID MLWiNS-2003002 and a Gift from Intel Co.}
}

\author{Ahmet~Faruk~Saz,~\IEEEmembership{Student Member,~IEEE,}
        Siheng~Xiong,~\IEEEmembership{Student Member,~IEEE,}
        and~Faramarz~Fekri,~\IEEEmembership{Fellow,~IEEE}
}

\maketitle

\begin{abstract}

This work presents an analysis of semantic communication in the context of First-Order Logic (FOL)-based deduction. Specifically, the receiver holds a set of hypotheses about the State of the World (SotW), while the transmitter has incomplete evidence about the true SotW but lacks access to the ground truth. The transmitter aims to communicate limited information to help the receiver identify the hypothesis most consistent with true SotW. We formulate the objective as approximating the posterior distribution of the transmitter at the receiver. Using Stirling’s approximation, this reduces to a constrained, finite-horizon resource allocation problem. Applying the Karush-Kuhn-Tucker conditions yields a truncated water-filling solution. Despite the problem’s non-convexity, symmetry and permutation invariance ensure global optimality. Based on this, we design message selection strategies, both for single- and multi- round communication, and model the receiver’s inference as an $m$-ary Bayesian hypothesis testing problem. Under the Maximum A Posteriori (MAP) rule, our communication strategy achieves optimal performance within budget constraints. We further analyze convergence rates and validate the theoretical findings through experiments, demonstrating reduced error over random selection and prior methods.

\end{abstract}

\begin{IEEEkeywords}
First-order logic, semantic communication, water-filling, hypothesis testing, deduction

\end{IEEEkeywords}

\section{Introduction}

\thispagestyle{empty}

Semantic communication redefines data transmission by prioritizing meaning over raw data, offering critical benefits for intelligent services and next-generation networks like 6G. Previous works have explored semantic compression techniques \cite{Gund_Yen}, knowledge-enhanced transmission strategies \cite{Qin2021SemanticCP}, and semantic-aware resource allocation. These efforts have shown the promise of semantic communication for improving efficiency and robustness across diverse intelligent systems and communication environments.

While these works largely emphasize semantic compression and transmission optimization, our focus shifts to a distinct but complementary problem: logical deduction of hypotheses under resource constraints. In many real-world scenarios—such as distributed sensing or privacy-sensitive environments—the transmitter lacks knowledge of the receiver’s precise goal. Consequently, the challenge becomes transmitting the most semantically content-informative message possible. This framework has the potential to empower a variety of applications involving logical deduction: drones in search-and-rescue operations can share logical summaries of their local observations \cite{Javaid2024LargeLM}; wearable health monitors can transmit structured symptom patterns to support remote diagnostics \cite{Mahmud}; autonomous vehicles can communicate compact causal cues like “reduced traction ahead” \cite{Mahmud}; and legal or educational systems can tailor logic-based explanations to user intent \cite{Singamaneni2023ASO}.

This paper builds on prior work that introduced a First-Order Logic (FOL)-based semantic communication framework for deductive inference. Drawing from the foundational ideas of Carnap \cite{c2} and Hintikka \cite{c10}, we treated FOL as a formal language for representing and transmitting structured knowledge. While classical semantic systems were computationally intractable, our earlier work \cite{saz2024lossy, saz2024model, saz2025discd} addressed scalability by framing communication as a boolean-satisfiability model selection problem. Using model counting, we developed a bit-efficient, task-agnostic transmission strategy that enabled logical inference at the receiver, even when the transmitter lacked access to the downstream task.

In this paper, we take our efforts a step further and formalize semantic communication as a posterior distribution alignment problem over a logical space structured by First-Order Logic (FOL). The world is represented using a hierarchy of constituents, attributive constituents, and Q-sentences, and uncertainty is quantified via inductive logical probability over these structures. The transmitter, observing partial evidence about the true state, aims to select messages that maximize semantic informativeness without access to the receiver’s hypothesis set. We show that optimal message selection reduces to minimizing the KL divergence between transmitter and receiver posteriors over constituents, which in turn becomes a constrained resource allocation problem over observed attributive constituents. To make the analysis tractable, we apply Stirling’s approximation and exploit the symmetry of the resulting objective, deriving a globally optimal solution using a truncated water-filling algorithm. We extend this framework to multi-round communication, showing that repeated optimal allocations converge to an ideal posterior configuration under cumulative constraints. This formulation enables principled message selection under realistic communication budgets and, under the Bayesian MAP criterion at the receiver, minimizes error in hypothesis deduction. We derive the convergence rate for the multi-round setup and show that the algorithm converges linearly. Finally, we validate our theoretical findings on a custom deduction dataset, demonstrating superior performance in hypothesis selection compared to prior work, greedy selection, and random baselines.
      
The contribution of this paper is twofold. First, it introduces a formal connection between semantic communication and logical inference by framing message transmission as the problem of posterior distribution alignment between transmitter and receiver in a First-Order Logic framework. This perspective distinguishes the present work from prior approaches that emphasize semantic compression or task-specific heuristics, and provides a principled way to quantify semantic informativeness through distributional divergence. Second, we develop and analyze a resource allocation formulation of this problem, showing that the optimal message composition follows a truncated water-filling solution. We establish that this strategy achieves global optimality in the relaxed setting and extend it to multi-round communication, where convergence properties are derived. Simulation results confirm that aligning posterior distributions yields significant improvements in hypothesis deduction accuracy compared to prior semantic communication strategies. Together, these contributions advance semantic communication beyond compression-oriented formulations and provide a theoretically grounded approach for logic-based deductive reasoning under communication constraints. In the next section, we introduce FOL representation of the world and problem setup.

\section{Related Work}

Semantic and goal-oriented communication systems, inspired by Weaver’s Level B notion, focus on transmitting meaning rather than symbols. DeepSC and DeepJSCC \cite{xie2021deep, bourtsoulatze2019deep} pioneered autoencoder-based designs, while other works introduced task-oriented training \cite{shao2022task, yang2023task} and functional compression \cite{Saidutta2022AML}. These methods achieve notable gains in efficiency but often rely on black-box training or static task definitions.  

Beyond end-to-end learning, neuro-symbolic and graph-based approaches aim to combine logical reasoning with semantic transmission. 
Prior work \cite{xiongtilp, yang2024temporal, xiong2024teilp, xiong2024large, xiong2024deliberate} has shown that these methods strengthen semantic understanding by capturing structured meaning, making them promising foundations for communication systems.
RuleTaker and ProofWriter \cite{clark2020transformers, tafjord2021proofwriter} demonstrated transformer-based reasoning, while recent studies integrate knowledge graphs and large models \cite{guo2023kg, nour2024semantic} to encode structured semantics. Vector quantization \cite{fu2023vector} and generative models \cite{grassucci2023generative} further extend semantic communication by aligning semantic tokens with modulation schemes or reconstructing content via shared AI models.

Our work differs in that it does not rely on neural compression or task-specific encoders. Instead, we adopt a symbolic and probabilistic formalism grounded in First-Order Logic and inductive probabilities tailored to logical inference. This allows us to frame semantic communication for hypothesis deduction as a posterior distribution alignment problem, derive a globally optimal truncated water-filling allocation, and establish convergence guarantees. Thus, our contributions provide a theoretical complement to learning-driven and application-specific approaches in semantic communication.

\section{Background on FOL, Inductive Probabilities, and Semantic Information}

\subsection{Representing the State of the World via First-Order Logic}

We begin by illustrating how the state of the world can be represented using First-Order Logic (FOL) following the conventions presented in \cite{c4, nelte}. Consider a language $\Language$ whose syntax is defined over vocabulary consisting of a finite set of dyadic predicate symbols $\Predicate = \{P_i\}_{i=1}^{k}$, a countably infinite set of variables $\Variable = \{x_i\}_{i \in \mathbb{N}}$, a countable set of constant symbols $\Constant = \{\top, \bot, ...\}$, the logical connectives $\{\neg, \land, \lor\}$, the quantifiers $\{\forall, \exists\}$, and parentheses $()$ for disambiguation. 

The syntax of a formal language has no inherent meaning on its own; meaning is provided through a \emph{model}. A model for a language ($\Language$) is a pair $\Model=\langle \Entity, J\rangle$, where $\Entity$ is a nonempty, countably infinite set called the domain (or universe) whose elements are referred to as entities, and $J$ is an interpretation function that assigns constant symbols to entities in $\Entity$ and each binary predicate symbol $P \in \Predicate$ a binary relation $ J(P) \subseteq \Entity \times \Entity$. 

To evaluate formulas that contain variables, a valuation function $v$ assigns domain elements to variables $v[x_i \mapsto e_i]$. Given a model and a valuation, the satisfaction relation ($\Model, v \models F$) formally specifies whether a formula ($F$) is true in the model $\Model$ under the corresponding valuation $v$. The \emph{state of the world} is then represented through true but unknown model-valuation pair $\mathcal M^*, v^*$, which exhaustively determines the structure of the domain $\Entity$ and the complete pattern of relationships among its elements, thereby logically entailing all correct statements of $\Language$.

As an initial step towards characterizing the state of the world in the framework explained, we start by defining specific predicate-relational configurations between entities through \textit{Q-sentences}.

\begin{definition}[Q-sentence]
For a fixed configuration $i$, let
$\delta_{i,P},\, \delta'_{i,P} \in \{+1,-1\}, \forall\, P \in \Predicate $.
For $x_1, x_2 \in \Variable$, define the \textit{Q-sentence} $Q_i(x_1, x_2)$ by
\[
Q_i(x_1, x_2) = \bigwedge_{P \in \Predicate}
\left( \delta_{i,P} P(x_1, x_2) \land \delta'_{i,P} P(x_2,x_1)\right).
\]
\end{definition}

\begin{remark}
The coefficients $\delta_{i,P}, \delta'_{i,P} \in \{+1,-1\}$ encode whether the
predicate $P$ appears positively or negatively in the Q-sentence. Concretely,
$\delta_{i,P} = +1$ corresponds to the atomic formula $P(x_1, x_2)$, while
$\delta_{i,P} = -1$ corresponds to its negation $\neg P(x_1, x_2)$; analogously,
$\delta'_{i,P}$ governs the occurrence of $P(x_2,x_1)$. Thus, each configuration $i$
specifies a complete choice of affirmation or negation for every predicate in
both argument orders.
\end{remark}

This definition implies that there are $4^{|\Predicate|}$ distinct Q-sentences. For a fixed valuation $v[x_1 \mapsto e_1, x_2 \mapsto e_2]$, exactly one of these $4^{|\Predicate|}$ Q-sentences evaluates to 1 (true), while all others evaluate to 0 (false). Building on this observation, we introduce the notion of an attributive constituent, which captures the complete relational profile of a given entity.

\begin{definition}[Attributive Constituent]
Let $\QSent$ be the set of all Q-sentences. For a fixed configuration $j$, let
$\delta_{j,Q} \in \{+1,-1\}, \forall\, Q \in \QSent$.
For $x_1 \in \Variable$, an \textit{attributive constituent} $C_j(x_1)$ is defined as:
\begin{equation}
    C_j(x_1) = \bigwedge_{Q \in \QSent} \delta_{j,Q}\bigl((\exists x_2)\, Q(x_1,x_2)\bigr)
\end{equation}
\end{definition}

\begin{remark}
As before, the coefficients $\delta_{j, Q} \in \{+1,-1\}$ encode polarity:
$\delta_{j, Q} = +1$ denotes the formula $(\exists x_2)Q(x_1, x_2)$, while
$\delta_{j, Q} = -1$ denotes its negation $\neg (\exists x_2)Q(x_1, x_2)$. $\neg (\exists x_2) Q(x_1, x_2)$ asserts that there exists no $x_2$ for which $Q(x_1, x_2)$ evaluates to true.
\end{remark}

This expression encapsulates the complete set of relationships that an entity \(x_1\) holds within the world, thereby characterizing \emph{what kind of individual} \(x_1\) is with respect to the relationships it possesses or lacks with others. Specifically, \(C_j(x_1)\) specifies the pattern of \(Q\)-relations that the individual \(x_1\) has with all other entities in the universe. It therefore explicitly identifies which Q-sentences, among the \(4^{|\Predicate|}\) possible ones, are satisfied by \(x_1\), as well as which are not.

For every individual \(x_1\) in the universe, there exists a corresponding attributive constituent \(C_j(x_1)\). While distinct individuals may share the same attributive constituent, the total number of distinct attributive constituents is bounded above by \(2^{4^{|\Predicate|}}\) by definition. 

To represent the state of the world in $\Entity$, we define \textit{constituents} as combinations of attributive constituents.

\begin{definition}[Constituent]
Let $\AttConst$ be the set of all attributive constituents. For a fixed configuration $k$, let
$\delta_{k,C_A} \in \{+1,-1\}, \quad \forall\, C_A \in \AttConst$. A \textit{constituent} \( C^\Width \) of width \( \Width \) in \( \C \) is a conjunction of attributive constituents:
\begin{equation}
    C^{\Width, k} = \bigwedge_{ C_A \in \AttConst} \delta_{k, C_A}\bigl((\exists x_1)C_A(x_1)\bigr),
\end{equation}
The width \( \Width \) is the number of attributive constituents in the conjunction with positive polarity $\delta_{k,C_A} = +1$.
\end{definition}

\begin{remark}
As before, the coefficients $\delta_{k,C_A} \in \{+1,-1\}$ encode polarity:
$\delta_{k,C_A} = +1$ denotes the formula $(\exists x_1)C_A(x_1)$, while
$\delta_{k,C_A} = -1$ denotes its negation $\neg (\exists x_1)C_A(x_1)$. \(\neg (\exists x_1)\, C_A(x_1)\) asserts that there exists no \(x_1\) for which \(C_A(x_1)\) evaluates to true.
\end{remark}

This definition implies that there are at most \(2^{2^{4^{|\Predicate|}}}\) distinct constituents. The constituents \(C^{\Width,k}\), for \(k = 1, \ldots\), enumerate all normal forms that can be constructed from the primitive predicate set \(\Predicate\), propositional connectives, and at most two quantifiers. Each constituent explicitly specifies which attributive constituents in \(\AttConst\) are realized in the universe, thereby fully characterizing the relationship patterns that exist in \(\Entity\).

Although up to \(2^{2^{4^{|\Predicate|}}}\) constituents can be defined for a universe \(\Entity\), exactly one of them is satisfied. Let $\mathcal C$ be the set of all constituents that can be defined in $\Language$ for the universe $\Entity$. For any true model--valuation pair \((\mathcal{M}^*, v^*)\) over \(\Entity\), there exists a unique constituent \(C^* \in \mathcal{C}\) such that \(\mathcal{M}^* \models C^*\). This constituent represents the \emph{true state of the world}.

We therefore interpret constituents as potential candidates for the state of the world, with the satisfied constituent corresponding to the actual state. In practice, determining the true state of the world thus amounts to identifying the constituent that is most consistent—i.e., best supported—by the available evidence \(\mathcal{E}\).

\begin{remark}
Multiple notions of the state of the world are possible, and the definition adopted above represents only one such formulation. In this work, we define the state of the world exclusively in terms of the relationship patterns that exist in \(\Entity\). Alternative formulations are possible; for example, one may define the state of the world by characterizing, for each individual \(x_1\), the type of individual \(x_1\) is with respect to the relationships it holds with others (i.e., as the set of attributive constituents \(\AttConst\) that holds true), or by adopting hybrid definitions that combine both perspectives, as in \cite{c4}.
\end{remark}

Every general sentence (and hence every general hypothesis) $\varphi$ of $\Language$ over $\Entity$ can be expressed as a disjunction of constituents, as established in the theorem below.

\begin{theorem}[Distributive Normal Forms of General Sentences in FOL \cite{nelte, c9, c4}]\label{const_dis}
Let \(\Language\) be a first-order language with a finite set of dyadic predicate symbols \(\Predicate\), and let every model \(\mathcal M = \langle \Entity, J \rangle\) for \(\Language\) have a countably infinite domain \(\Entity\). Let $\C$ be the set of constituents described in $\Language$. For any closed first-order general sentence \(\varphi \in \Language\), consider the set
\[
\{\, C^{w,k} \in \C \;:\; C^{w,k} \models \varphi \,\}
\]

such that
\begin{equation}\label{eq:constituent-disjunction}
\varphi \;\equiv\;
\bigvee_{\substack{C^{w,k} \in \C \\ C^{w,k} \models \varphi}} C^{w,k}.
\end{equation}

For every model \(\mathcal M\), there exists a unique constituent \(C^{\mathcal{M}} \in \C\) such that \(\mathcal M \models C^{\mathcal{M}}\). Then, for every model \(\mathcal M\), $ \mathcal M \models \varphi \quad \text{if and only if} \quad \mathcal M \models C^{\mathcal{M}} \text{ for some constituent } C^{\mathcal{M}} \in \C \text{ with } C^{\mathcal{M}} \models \varphi$.
\end{theorem}

In the first-order logic (FOL)–based representation of the world \(\Entity\), we introduce a probability measure \(\ProbMeas\) over the constituents \(C^{\Width,k}\). This measure is defined using \emph{inductive logical (also known as epistemic) probabilities}, which function as inductive Bayesian posteriors constructed from observed evidence and, when available, a prior distribution.

\begin{definition}[Inductive Logical Probability \cite{c2}]
For any statement \(\varphi\) and evidence \(e\) (observations) in \(\Entity\), the \emph{inductive logical probability} (also referred to as the \emph{degree of confirmation}) \(c(\varphi, e)\) of \(\varphi\) given \(e\) is defined as
\begin{equation}\label{c_func}
c(\varphi, e) \;=\; p(\varphi \mid e) \;=\; \frac{p(\varphi \land e)}{p(e)}.
\end{equation}
\end{definition}

We then define an inductive logical probability measure \(\ProbMeas\) over the set of constituents \(\mathcal{C}\) as
\[
\ProbMeas \;=\; \bigl\{\, c\!\left(C^{w,k}, e\right) \;\big|\; C^{w,k} \in \mathcal{C} \,\bigr\},
\]
where \(c(C^{w,k}, e)\) denotes the inductive logical probability of the constituent \(C^{w,k}\) given the evidence \(e\), and
\[
\sum_{C^{w,k} \in \mathcal{C}} c\!\left(C^{w,k}, e\right) \;=\; 1.
\]

Once the probabilities of the constituents are fixed, the probability—and hence the inductively determined truth—of any hypothesis can be computed in accordance with Theorem~(1).

It is important to note that the assignment of probabilities to constituents is relative to the inductive system adopted. Different systems—such as Carnap’s $c$-system~\cite{carnap_logical_1962}, the $\lambda$-continuum~\cite{c1}, and Hintikka’s two-dimensional framework~\cite{c10}—employ different ways of determining the probability assignments~\cite{c1}. Consequently, these systems induce different probability distributions over constituents. In the present work, we adopt the conventions and framework developed in~\cite{c4}.

To quantify the remaining uncertainty after incorporating new observations, Carnap introduced the notion of semantic \emph{content-information}, defined as follows:

\begin{definition}
The \textit{cont-information} of a sentence quantifies the informativeness of an observation \( e \) for a given observer in a world \( \Entity \) by measuring the reduction in uncertainty about a general sentence (i.e., a general hypothesis) \( \varphi \) of the language \( \Language \).  
Given a sentence \( \varphi \in \Language \) and evidence \( e \) (observations) in \( \Entity \), the cont-information of \( \varphi \) given \( e \) is defined as
\begin{equation}
\text{cont}(\varphi; e) \; \coloneqq \; 1 - c(\varphi, e) \;=\; 1 - p(\varphi \mid e).
\end{equation}
\end{definition}

In intuitive terms, $\text{cont}(\varphi; e)$ measures the reduction in uncertainty about the truth of a sentence $\varphi$ after observing the evidence $e$. In the following, the cont-information measure is used to develop a semantic communication framework.

\section{Node-to-Node Semantic Communication and Deductive Inference of Hypothesis}

We consider a node-to-node semantic communication setting in which a transmitter node \( N_{Tx} \) communicates with a receiver node \( N_{Rx} \). The receiver's objective is to identify the most plausible hypothesis \( h^* \in \mathcal{H} = \{ h_i \mid i \in \mathcal{I}_\mathcal{H} \} \), where each hypothesis \( h_i \) is a self-consistent FOL statement expressible as in Thm. ~\eqref{const_dis} and $\mathcal{I}_\mathcal{H}$ is the index set over hypotheses. The true world state corresponds to an unknown constituent \( C^* \in \mathcal{C} \), and the receiver infers \( h^* \), the hypothesis most consistent with $C^*$ based on messages \( \{\varphi_t\}_{t=1}^T \) received from the transmitter over $T$ rounds.

Previously, we introduced SCLD \cite{saz2024lossy}, a cont-information based message selection framework. At each round \( t \), the transmitter, using its local observations \( \Evidence_t \), selects the semantically most content-informative FOL message \( \varphi_t \) of bounded size \( |\varphi_t| \leq B \), aiming to reduce the receiver's uncertainty over \( \mathcal{C} \) for better deduction of $h^*$. The transmitter does not know anything about the hypothesis set \( \mathcal{H} \). As such, in our framework, the transmitter constructs and sends the message $\varphi_t$ at round \( t \) as:

\begin{equation}
    \varphi_t^* = \argmax_{\varphi \in \Evidence_t \setminus \{\varphi_1, \dots, \varphi_{t-1}\}} \text{cont}\left(\varphi, \Evidence_r \land \bigwedge_{z = 1}^{t-1} \varphi_z\right), \text{s.t. } |\varphi| \leq B,
\end{equation}
where $\Evidence_t$ is the evidence observed by the transmitter and $\Evidence_r$ is the evidence present at the receiver about the SotW.

Another approach to semantic communication for hypothesis inference involves accurately conveying the underlying distribution through messages $\varphi$. In other words, the receiver's objective becomes replicating the sender's intended distribution as precisely as possible. This accuracy requirement can be quantitatively defined using the Kullback--Leibler (KL) divergence between the posterior distributions at the transmitter, \( p(C^{w,k}_t \mid \Evidence_t) \), and at the receiver, \( p(C^{w,k}_r \mid \Evidence_r) \):
\begin{equation}
\mathrm{D}_{\mathrm{KL}}\!\left( p_t \,\|\, p_r \right) 
= \sum_{w = f}^K p_t(w) \log \frac{p_t(w)}{p_r(w)}.
\end{equation}
The transmitter then selects $\varphi$ to minimize this divergence.

Upon receiving \( \varphi_t^* \) at round $t$, the receiver updates its inductive logical probability distribution \( \ProbLoc{I} \) over constituents, refining its epistemic beliefs. After \( T \) rounds, the receiver selects the most probable hypothesis as:
\begin{equation}\label{optim_comm}
    h^* = \argmax_{h \in \mathcal{H}} c\left(h, \bigwedge_{t = 1}^{T} \varphi_t \right).
\end{equation} 
  
To analyze our framework, we consider finite evidence $\Evidence$ containing \( f \) distinct attributive constituents about the SotW, and let $n_j$ be the number of attributive constituents corresponding to $C_j(x)$, i.e., the number of times $n_j$ that has shown up in set $\Evidence$. A constituent \( C^{w,k} \) of width \( w \) is compatible with \( \Evidence \) if \( f \leq w \leq \NumAttConst \), where $f$ is the width of the minimal constituent and \( \NumAttConst \) is the total number of possible attributive constituents. The posterior over \( \mathcal{C} \), the set of constituents, is given by:   
\begin{equation} \label{eq:posterior-dist}
    \ProbLoc{I} = \{c(C^{w,k}, \Evidence) \mid f \leq w \leq \NumAttConst \},
\end{equation}
with each term derived via Bayes' rule:
\begin{equation} \label{eq:posterior}
    p(C^{w,k} \mid \Evidence) = \frac{p(C^{w,k}) p(\Evidence \mid C^{w,k})}{\sum_{w'} p(C^{w',k}) p(\Evidence \mid C^{w',k})},
\end{equation}

If the prior \( p(C^{w,k}) \) and likelihood \( p(\Evidence \mid C^{w,k}) \) are selected as in \cite{c10}, the posterior $p(C^{w,k} \mid \Evidence)$ becomes:

\begin{align} \label{posti}
&p(C^{w,k} \mid \Evidence) = \\
&\parbox{0.9\linewidth}{%
\[
\frac{
    \pi(\alpha, w\lambda / K) \prod_{j=1}^{f} \pi(n_j, \lambda / w)
}{
    \sum_{i=0}^{K - f}
    \left[
        \binom{K - f}{i}
        \pi\left(\alpha, (f + i)\lambda / K\right)
        \prod_{j=1}^{f} \pi\left(n_j, \lambda / (f + i)\right)
    \right]
}
\]
} \nonumber 
\end{align}

\noindent where, $\alpha \geq 0$ is the inductive generalization coefficient. $\pi(\alpha, x)$ is a Pochhammer symbol (shifted factorial) defined as:
\begin{equation}\label{eqn12}
\pi(\alpha, x)  = \prod_{i = 0}^{\alpha - 1} (i + x) = \frac{\Gamma(x + \alpha)}{\Gamma(x)}.
\end{equation}
\noindent Here, $\Gamma(x)$ is the Gamma function. The expression in \eqref{posti} involves ratios of these Pochhammer symbols (and thus Gamma functions), making it complex.

However, this form is amenable to asymptotic analysis using Stirling’s approximation, which describes the behavior of the Gamma function for large arguments: $ \Gamma(z) \approx \sqrt{2\pi} \, z^{z - 1/2} e^{-z} $ for $ z \gg 0$. This approximation is generally accurate for $z \geq 1$ and very close for $z \geq 10$. Assuming large $\alpha$ (i.e., $\alpha \geq 1$) and/or large observation counts $n_j$ (i.e., $n_j \geq 1$), we can apply Stirling's approximation to the Gamma function terms within the Pochhammer symbols in \eqref{posti}.

After applying the approximation and focusing on the dominant terms (where many constants and exponential factors cancel out), the posterior $p(C^{w,k} \mid \Evidence)$ simplifies considerably. The key dependencies involving the counts $n_j$ become apparent in the simplified form:

\begin{align} \label{posti_approx}
&p(C^{w,k} \mid \Evidence) \approx \\
&\small \parbox{\linewidth}{%
\[
\frac{1}{
\sum_{i=0}^{K - f}
\binom{K - f}{i}
\frac{\Gamma\left( \frac{w\lambda}{K} \right)}{\Gamma\left( \frac{(f+i)\lambda}{K} \right)}
\alpha^{\frac{(f+i - w)\lambda}{K}}
\left(
\frac{\Gamma\left( \frac{\lambda}{w} \right)}{\Gamma\left( \frac{\lambda}{f + i} \right)}
\right)^f
\prod_{j=1}^{f} n_j^{\left(\frac{1}{f+i} - \frac{1}{w} \right)\lambda}
}
\]
} \nonumber 
\end{align}

\noindent where $\alpha \geq 1$ and $\lambda \geq 0$ are constant coefficients determining the impact of priors in singular inductive inferences and inductive generalizations \cite{c10}. Next, we formulate an optimization objective to determine an optimal method of semantic communication for hypothesis deduction.

\section{Single Round Semantic Communication as Distribution Approximation}

In formal semantics grounded in first-order logic (FOL), world representations possess a well-defined structure. This structure, as declared in (1), (2), (3), and (4), permits any FOL statement including the hypothesis to have a well-defined probability assignment. Under this framework, the semantics of a hypothesis and it's degree of confirmation can be determined. Hence, the optimal strategy for semantic communication for hypothesis deduction can be naturally formulated as the accurate transfer of this distribution. That is, the ultimate goal is for the receiver to reconstruct the intended distribution as accurately as possible. This fidelity criterion can be formalized using the Kullback–Leibler (KL) divergence between the transmitter \( p_t = p(C^{w,k}_t \mid \Evidence_t) \) and receiver \( p_r = p(C^{w,k}_r \mid \Evidence_r) \) posterior distributions:

\begin{equation}
    \mathcal{L}_{\text{KL}} = \mathrm{D}_{\mathrm{KL}}\left( p_t \,\big\|\, p_r \right) = \mathbb{C}(p_t, p_r) - \mathbb{H}(p_t)
\end{equation}

where \( \mathbb{H}(\cdot) \) is Shannon entropy and \( \mathbb{C}(\cdot) \) is cross-entropy. Since \( \mathbb{H}(p_t) \) is constant w.r.t. the transmission, minimizing \( \mathcal{L}_{\text{KL}} \) is equivalent to minimizing \( \mathbb{C}(p_t, p_r) \), or (after simplifying logarithm) maximizing:
\begin{equation} \label{eq:max_objective_condensed}
\max_{\{n_j'\}} \sum_{w=f}^{K} p_t(w) \cdot p_r(w) \quad \text{s.t.}
\end{equation}
\begin{equation} \label{eq:constraints}
\sum_{j=1}^f n_j' = B, \quad 0 \leq n_j' \leq n_j, \quad n_j' \in \mathbb{Z}_{\geq 0}
\end{equation}
Here, \( \{n_j'\} \) are the counts of the \( f \) distinct attributive constituent types, as determined by the selection of message $g$, to be transmitted to the receiver (forming \( \Evidence_r \)), \( B \) is the total transmission budget (e.g., total number of attributive constituents communicated), and \( n_j \) is the available number of attributive constituent \( C_j(x) \) at the transmitter (from \( \Evidence_t \)). The transmitter distribution \( p_t(w) = p(C^{w,k}_t \mid \Evidence_t) \) is fixed throughout communication as no new evidence is observed. The receiver distribution \( p_r(w) = p(C^{w,k}_r \mid \Evidence_r) \) depends on the chosen \( \{n_j'\} \) via the approximate posterior relation derived from \eqref{posti_approx}:
\begin{equation} \label{eq:pr_dependency}
p_r(w) \propto \left[ \sum_{i=0}^{K-f} A_w(i) \cdot \prod_{j=1}^f (n_j')^{\beta_i(w)} \right]^{-1}
\end{equation}
with \( \beta_i(w) := \left( \frac{1}{f+i} - \frac{1}{w} \right)\lambda \). Notably, \( \beta_i(w) \) is independent of the constituent type index \( j \), making the objective function \eqref{eq:max_objective_condensed} symmetric in the variables \( \{n_j'\} \).

This is a constrained resource allocation problem (of distributing the limited transmission budget $B$ across possible attributive constituents to best approximate the transmitter’s belief distribution) under integer constraints, typical in communication systems where discrete packets or symbols must be transmitted. To find the optimal allocation, we analyze its continuous relaxation (treating \( n_j' \) as real, temporarily ignoring \( n_j' \leq n_j \)). The Lagrangian is:
\begin{align}
\mathcal{L}(\{n_j'\}, \mu) &= \sum_{w=f}^K p_t(w) \cdot p_r(w) - \mu \left( \sum_{j=1}^f n_j' - B \right)
\end{align}
The Karush-Kuhn-Tucker (KKT) conditions require \( \frac{\partial \mathcal{L}}{\partial n_k'} = 0 \) for all \( k=1, \dots, f \). Computing the derivative yields a system of equations relating \( \{n_j'\} \) and the Lagrange multiplier \( \mu \). Crucially, due to the symmetry of the objective function \eqref{eq:pr_dependency} and the budget constraint \( \sum n_j' = B \), the expressions for \( \frac{\partial \mathcal{L}}{\partial n_k'} \) are identical for all \( k \). Therefore, the unique solution satisfying all KKT conditions simultaneously must be symmetric:
\begin{equation} \label{eq:symmetric_sol}
n_j'^* = \frac{B}{f}, \quad \text{for all } j = 1, \dots, f
\end{equation}
Given the symmetry, this stationary point corresponds to the global optimum of the relaxed problem \cite{boyd2004convex}. This suggests evenly distributing the transmission budget \( B \) across all \( f \) relevant semantic constituent types when capacities are unlimited.

In practice, we must respect the integer nature of counts and the per-attributive-constituent type capacity limits \( n_j' \leq n_j \), and hence (19) cannot be directly implemented. Instead, incorporating these constraints transforms the solution into a \textit{truncated water-filling} allocation. To find the optimal integer allocation \( \{n_j'^*\} \) satisfying \eqref{eq:constraints}, the constituent types are first sorted based on available counts (\( n_{(1)} \leq \dots \leq n_{(f)} \)). Next, a water level \( \theta \) is determined such that \( \sum_{j=1}^f \min(\theta, n_j) = B \), yielding the optimal continuous allocation respecting capacities, \( n_j^{\text{cont}} = \min(\theta, n_j) \). Finally, the integer allocation \( \{n_j'^*\} \) is obtained by rounding \( \{n_j^{\text{cont}}\} \) (e.g., using a standard rounding procedure) such that \( \sum_{j=1}^f n_j'^* = B \) and \( n_j'^* \leq n_j \) for all \( j \). This algorithm provides the optimal message composition strategy for single-round semantic communication under budget and evidence capacity constraints.

\begin{lemma}[Rounding Error Bound]
Let $\{n_j^{\mathrm{cont}}\}$ denote the optimal continuous allocation obtained by truncated water-filling and $\{n_j'^*\}$ the corresponding integer allocation after rounding subject to $\sum_j n_j'^* = B$. Then the deviation satisfies
\begin{equation}
\max_{j} \left| n_j'^* - n_j^{\mathrm{cont}} \right| \leq 1,
\end{equation}
and consequently the cumulative rounding error is bounded as
\begin{equation}
\left\| \mathbf{n}'^* - \mathbf{n}^{\mathrm{cont}} \right\|_1 \leq f.
\end{equation}
\end{lemma}

\begin{proof}
Since rounding each component to the nearest integer introduces an error of at most $1$, and the budget constraint ensures that positive and negative deviations must balance across $f$ components, the $\ell_\infty$ bound follows directly, with the $\ell_1$ bound obtained by summation.
\end{proof}

\begin{remark}[Symmetry Assumption]
The global optimality of the symmetric allocation in~\eqref{eq:symmetric_sol} follows under the assumption that the feasible set is permutation invariant, i.e., the upper bounds on $n_j'$ are identical across all $j$. When heterogeneous bounds $n_j$ are present, the continuous stationary point remains symmetric but may no longer lie in the feasible set. In such cases, the optimal allocation is obtained via truncated water-filling, which respects individual capacity limits while preserving the symmetry argument within the feasible region.
\end{remark}

\section{Multi-Round Semantic Communication}

In many practical systems, communication occurs over multiple rounds. We extend our single-round water-filling framework to a multi-round setting with \( T \) rounds, where each round \( t \) has a transmission budget \( B \). The optimal strategy depends on when the receiver performs hypothesis deduction. We consider two primary scenarios:

\textbf{I. Long-term Optimization (Final Hypothesis Selection):}
If the receiver evaluates hypotheses only after all \( T \) rounds are completed, the transmitter's goal is to optimize the quality of the final, cumulative message set. The total budget available across all rounds is \( TB \). This scenario effectively reduces to the single-round optimization problem (Section V) but with an increased total budget.
The optimal cumulative allocation \( \{n_j'^{\text{total}}\} \) is found using the truncated water-filling algorithm with total budget \( TB \) and the initial available capacities \( \{n_j\} \). Specifically, find the water level \( \theta_{\text{total}} \) such that \( \sum_{j=1}^{f} \min(\theta_{\text{total}}, n_j) = TB \), and the corresponding allocation is \( n_j'^{\text{total}} = \min(\theta_{\text{total}}, n_j) \) (followed by integer rounding respecting constraints).

\textbf{II. Short-term Optimization (Per-Round Hypothesis Selection):}
If the receiver performs hypothesis deduction after each round \( t \), the transmitter must optimize the message sent in the current round to maximize the receiver's inference quality incrementally, considering the information already transmitted.
At each round \( t \) (\( 1 \leq t \leq T \)), the transmitter calculates the optimal allocation \( \{n_j'^{(t)}\} \) for the current round using the truncated water-filling algorithm with the per-round budget \( B \) and the remaining capacities \( \{r_j(t)\} \). The remaining capacity for constituent type \( j \) at the start of round \( t \) is:
\begin{equation}
r_j(t) = n_j - \sum_{\tau=1}^{t-1} n_j'^{(\tau)}
\end{equation}
where \( n_j'^{(\tau)} \) was the allocation for type \( j \) transmitted in round \( \tau \). The water level \( \theta_t \) for round \( t \) satisfies \( \sum_{j=1}^f \min(\theta_t, r_j(t)) = B \), and the allocation for round \( t \) is \( n_j'^{(t)} = \min(\theta_t, r_j(t)) \) (followed by integer rounding respecting the budget \( B \) and remaining capacities \( r_j(t) \)). This myopic strategy maximizes the immediate semantic gain at each round under the per-round budget constraint.

\section{Message Selection Algorithm from Candidates}

The strategies above assume the transmitter can construct the exact water-filling allocation. In practice, the transmitter might have a predefined set of \( M \) candidate messages, \( \mathbf{M} = \{\mathcal{A}^{(1)}, \dots, \mathcal{A}^{(M)}\} \). For example, these message can be identical to the M evidences that the transmitter has observed. Each candidate message \( \mathcal{A}^{(m)} = \{n_{j}^{(m)}\}_{j=1}^f \) represents a valid allocation of semantic constituents for a single round, satisfying \( \sum_{j=1}^f n_{j}^{(m)} = B \) and \( n_{j}^{(m)} \leq n_j \) (assuming capacities \( n_j \) are known or estimated). The transmitter must select \( T' \) (\( T' \leq T \)) messages from \( \mathbf{M} \) for transmission over \( T' \) rounds. This selection should approximate the ideal strategies outlined in Section VI.

Again, the selection algorithm depends on the optimization objective:

\textbf{a) Approximate Long-Term Optimization:}
To approximate the ideal long-term strategy, the transmitter selects a subset \( \mathcal{S} \subset \mathbf{M} \) of size \( |\mathcal{S}| = T' \) such that the final cumulative allocation \( \mathbf{N}^{\text{final}} = \sum_{m \in \mathcal{S}} \mathcal{A}^{(m)} \) is as close as possible to the ideal total water-filling allocation \( \mathrm{WF}^{\text{final}} \) calculated with budget \( T'B \). The objective is to find the subset \( \mathcal{S} \) that minimizes:
\begin{equation}
\left\| \sum_{m \in \mathcal{S}} \mathcal{A}^{(m)} - \mathrm{WF}^{\text{final}} \right\|
\end{equation}
subject to the constraint that the cumulative allocation respects capacities: \( \sum_{m \in \mathcal{S}} n_j^{(m)} \leq n_j \) for all \( j \). In this case, the order of message transmission within the \( T' \) rounds does not affect the final outcome. This is a subset selection problem, which can be complex but aims to match the overall target distribution optimally using the available message pool.

\textbf{b) Approximate Short-Term Optimization:}
To approximate the ideal per-round optimization strategy, the transmitter should select messages greedily over \( T' \) rounds. At each round \( t \) (\( 1 \leq t \leq T' \)), let the cumulative allocation sent so far be \( \mathbf{N}^{(t-1)} = \sum_{\tau=1}^{t-1} \mathcal{A}^{(\sigma(\tau))} \), where \( \sigma(\tau) \) is the index of the message sent in round \( \tau \). The ideal target cumulative allocation after \( t \) rounds, according to the ideal short-term strategy, is \( \mathrm{WF}^{(t)} \), representing the water-filling allocation with total budget \( tB \). The transmitter selects the next message \( \mathcal{A}^{(\sigma(t))} \) from the remaining candidates in \( \mathbf{M} \) (ensuring capacities are not exceeded, i.e., \( N_j^{(t-1)} + n_j^{(\sigma(t))} \le n_j \)) that minimizes the deviation from the ideal target for the current cumulative step, e.g., minimizing \( \|\mathbf{N}^{(t-1)} + \mathcal{A}^{(\sigma(t))} - \mathrm{WF}^{(t)} \| \). This involves selecting both the message and its transmission order \( \sigma \). A common heuristic is to greedily choose the message at round \( t \) that brings the current cumulative allocation \( \mathbf{N}^{(t)} \) closest to the ideal target \( \mathrm{WF}^{(t)} \).

\section{Convergence Analysis}

We now analyze the convergence properties of the practical message selection algorithms presented in Section VII, comparing their performance to the ideal water-filling allocations derived in Section VI. The analysis considers the two optimization objectives: minimizing cumulative error over rounds (short-term) and minimizing final error (long-term). Let $T'$ be the number of messages selected and transmitted.

\enlargethispage{-2\baselineskip}

\begin{lemma}[Convergence Rate for Short-Term Optimization Approximation]
Let $\mathbf{N}^{(T')} = \sum_{t=1}^{T'} \mathcal{A}^{(\sigma(t))}$ denote the cumulative allocation after $T'$ rounds using the greedy message selection algorithm (Section VII.a) which aims to approximate the ideal short-term strategy. Let $\mathrm{WF}^{(T')}$ be the ideal cumulative water-filling allocation targeting budget $T'B$. The maximum deviation between the achieved allocation and the ideal target satisfies:
\begin{equation}
\max_{j=1,\dots,f} \left| N_j^{(T')} - \mathrm{WF}_j^{(T')} \right| = \mathcal{O}\left( \frac{f}{T'} \right).
\end{equation}

\end{lemma}

\textit{Proof Sketch.} The greedy algorithm at each round $t$ selects a message $\mathcal{A}^{(\sigma(t))}$ to minimize the deviation of the current cumulative sum $\mathbf{N}^{(t)}$ from the ideal target $\mathrm{WF}^{(t)}$. By iteratively reducing the largest deviations across the $f$ constituent types using messages of bounded size ($B$), the overall maximum deviation after $T'$ rounds tends to decrease proportionally to $1/T'$. The factor $f$ accounts for balancing across all constituent types.

\begin{lemma}[Convergence Rate for Long-Term Optimization Approximation]
Let $\mathbf{N}^{\text{final}} = \sum_{m \in \mathcal{S}} \mathcal{A}^{(m)}$ be the final cumulative allocation achieved by selecting a subset $\mathcal{S}$ of $T'$ messages from the candidate pool $\mathbf{M}$ using the algorithm in Section VII.b, aiming to approximate the ideal long-term strategy. Let $\mathrm{WF}^{\text{final}}$ be the ideal water-filling allocation targeting budget $T'B$. The maximum deviation satisfies:
\begin{equation}
\max_{j=1,\dots,f} \left| N_j^{\text{final}} - \mathrm{WF}_j^{\text{final}} \right| = \mathcal{O}(1).
\end{equation}
\end{lemma}

\textit{Proof Sketch.} This algorithm selects the best fixed subset $\mathcal{S}$ of size $T'$ from the candidate pool $\mathbf{M}$ to match the final target $\mathrm{WF}^{\text{final}}$. The approximation error is primarily determined by the granularity of the candidate messages in $\mathbf{M}$ and how well any combination of $T'$ messages can match the specific target profile. Since the optimization only considers the final state and relies on a potentially limited discrete set of candidate messages, the error bound does not necessarily decrease with $T'$ and remains bounded by a constant factor reflecting the quality and diversity of the pool $\mathbf{M}$.

\section{Simulation Results and Discussion}
\label{sec:simulation}

\begin{figure}[htbp]
    \centering
    \includegraphics[width=0.95\linewidth]{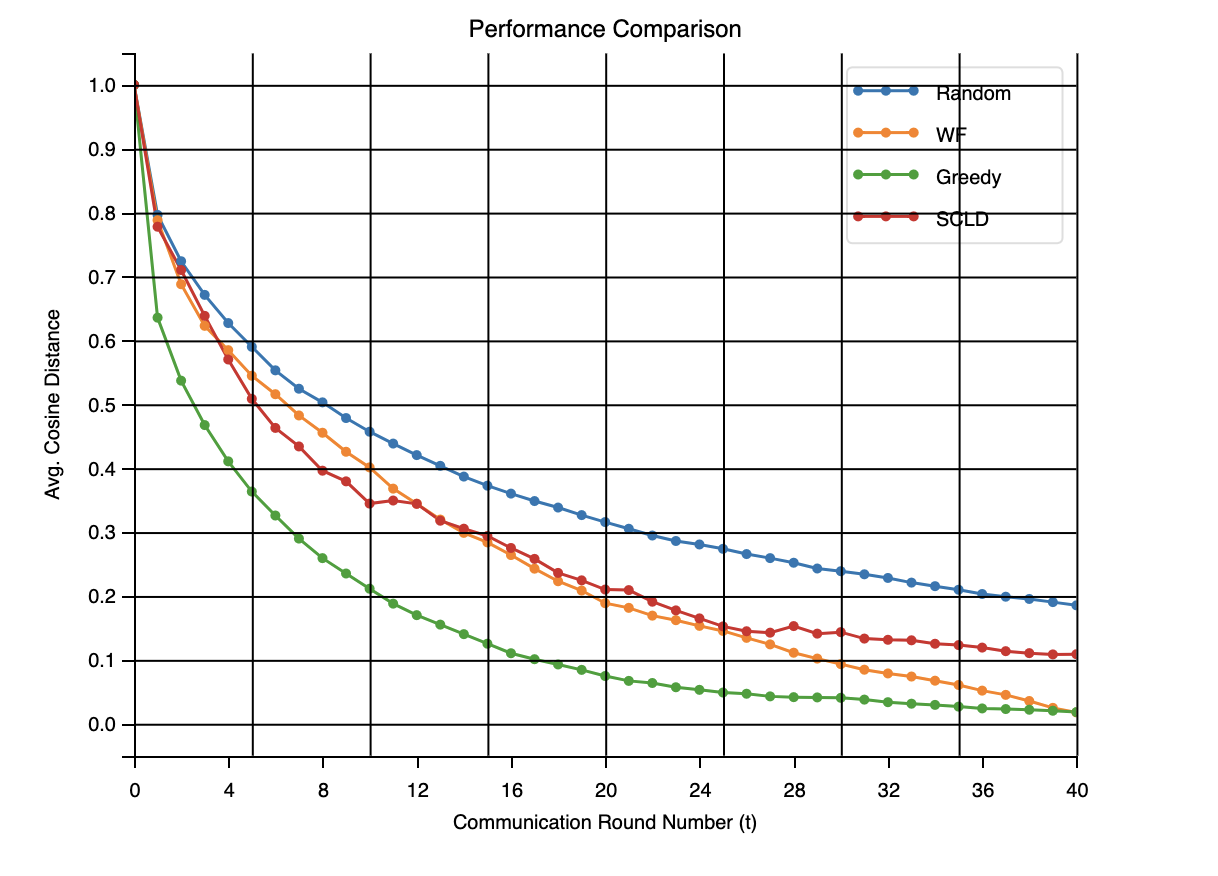} 
    \caption{Average Cosine Distance between receiver and transmitter distributions over $T=20$ rounds. Lower is better.}
    \label{fig:sim_distance}
\end{figure}

We evaluate the proposed semantic communication strategies through simulations designed to measure both the alignment between transmitter and receiver distributions and the resulting accuracy in hypothesis deduction. All derivations and simulations assume an error-free logical channel, i.e., transmitted constituents are received without distortion. This assumption isolates the effect of semantic resource allocation from channel impairments. In practice, symbol errors and partial misdecoding are unavoidable \cite{yang2024harnessing}, and their impact on the proposed allocation strategy remains an open question. Preliminary robustness checks under small error rates indicate graceful degradation, and a systematic analysis of noisy channels constitutes an important direction for future work.

The system consists of a transmitter that observes $N_{obs}$ attributive constituents drawn from a distribution over $K$ types, while the receiver begins with no prior evidence. Communication proceeds for $T$ rounds, each constrained by a budget $B$ constituents per message. Unless otherwise stated, parameters are set to $K \approx 1000$, $N_{obs}=750$, $T=20$, $B=5$, $\lambda=1.0$, and the hypothesis set size $K_{\mathrm{hypo}}=10$. For each run, one hypothesis is selected as the ground truth and the transmitter’s evidence is generated from its associated distribution. All results are averaged over $N_{\mathrm{sim}}=50$ independent runs.

Five strategies are compared. The Random-Free and Random-Chunk methods serve as baselines, selecting constituents either directly from the transmitter’s pool or from predefined evidence chunks. The SCLD approach follows the semantic cont-information framework by selecting messages that maximize immediate informativeness. Finally, the WF-Greedy and WF-Long strategies implement the proposed truncated water-filling allocation, the former optimizing myopically at each round and the latter computing a global plan over the total budget $TB$.

Performance is quantified using Cosine distance between the transmitter’s and receiver’s distributions. Although many similarity measures are available in the literature—including Kullback–Leibler divergence, Jensen–Shannon distance, Total Variation distance, and Hellinger distance—it is well established that these measures provide closely related trends when comparing probability distributions. Cosine distance therefore serves as a representative and computationally efficient proxy. Extending the evaluation to multiple metrics may nevertheless provide additional insights into the sensitivity of the algorithms, and constitutes an important direction for future work.

Figure~\ref{fig:sim_distance} shows the convergence of the receiver’s distribution toward the transmitter’s distribution over $T=20$ rounds. Both water-filling strategies significantly outperform the baselines, with WF-Greedy providing faster initial convergence and WF-Long achieving the lowest final divergence. WF-Greedy and WF-Long both reach $98\%$, compared to $89\%$ for SCLD and below $80\%$ for the random baselines. These results confirm that aligning posterior distributions leads directly to improved hypothesis inference.

Experiments further reveal that the relative differences in performance among the strategies depend on and vary with the communication budget $B$ and the number of samples $N_{obs}$. In some regimes, WF-Greedy provides stronger advantages, while in others WF-Long dominates. There are certain cases where random provides close accuracy to WF-Greedy and even surpassing WF-Long. This variability suggests that sensitivity analysis across a broader range of budgets, sample sizes, and evidence imbalance is necessary to fully characterize algorithmic performance. Moreover, although the present evaluation is limited to synthetic Dirichlet–multinomial environments, extending experiments to real-world reasoning datasets involving noise, incomplete evidence, and large-scale logical corpora and incorporating varying budget constraints will provide a more realistic assessment of the proposed framework. Despite these limitations, the initial results are highly promising, demonstrating the effectiveness of posterior alignment as a principle for semantic communication in hypothesis deduction.

\section{Conclusion}
\label{sec:conclusion}

This paper analyzed semantic communication for logic-based hypothesis deduction by framing message transmission as posterior distribution alignment between transmitter and receiver. We showed that the resulting constrained resource allocation problem admits a globally optimal truncated water-filling solution, and extended the framework to multi-round communication with established convergence properties. Simulations demonstrated that aligning posteriors significantly improves hypothesis inference, with both greedy and long-term strategies outperforming random and cont-information baselines under the MAP rule, while remaining robust to evidence imbalance and moderate noise. Although our study was limited to synthetic Dirichlet–multinomial environments and assumed noiseless communication, the promising results motivate future work on scalability to larger logical corpora, robustness under noisy or incomplete evidence, and integration with practical constraints in next-generation semantic communication.

\bibliographystyle{unsrt} 
\bibliography{refs, privacy, paper} 

\appendices

\section{Model Theoretic Foundations of Language $\Language$}

Consider a language $\Language$ whose syntax is defined over vocabulary consisting of a finite set of dyadic predicate symbols $\Predicate = \{P_i\}_{i=1}^{k}$, a countably infinite set of variables $\Variable = \{x_i\}_{i \in \mathbb{N}}$, a countable set of constant symbols $\Constant = \{\top, \bot, ...\}$, the logical connectives $\{\neg, \land, \lor\}$, the quantifiers $\{\forall, \exists\}$, and parentheses $()$ for disambiguation. Constant symbols and variables are called terms $\Term = \Variable \cup \Constant$. The set of \emph{formulas} of $\Language$ is defined inductively, beginning with atomic formulas and extending via logical connectives and quantifiers. In particular, if $P \in \Predicate$ and $t_1, t_2 \in \Term$, then $P(t_1,t_2)$ is an atomic formula. If $F,G$ are formulas, then $\neg F$, $F \land G$, and $F \lor G$ are also formulas. Finally, if $F$ is a formula and $x_i$ is a variable, then both $(\exists x_i)F$ and $(\forall x_i)F$ are formulas.

Syntax of the language carry no meaning independent of interpretation. Interpretation is introduced into $\Language$ by specifying \emph{models} for $\Language$. A \emph{model} for $\Language$ is a pair $\Model = \langle \Entity, J \rangle$, 
where $\Entity = \{e_i\}_{i \in \mathbb{N}}$ is a nonempty, countably infinite set called the \emph{domain} (or universe), whose elements are referred to as \emph{entities}. $J$ is an interpretation function assigning to each constant symbol $c_i \in \Constant$ entities $J(c_i) \subseteq \Entity$, to each binary predicate symbol $P \in \Predicate$ a binary relation $ J(P) \subseteq \Entity \times \Entity$.

A \emph{valuation} is a function $v : \Variable \to \Entity$ assigning elements of the domain to the variables of the language. Given a model $\Model = \langle \Entity, J \rangle$ and a valuation $v$, the satisfaction relation $\Model, v \models F$ determines whether a formula $F$ holds in $\Model$ under the variable assignment specified by $v$.

In particular, an atomic formula $P(t_1,t_2)$ is satisfied in $\Model$ under a valuation $v$ if and only if $(v(t_1), v(t_2)) \in J(P)$. Similarly, $(\exists x)F$ is satisfied in $\Model$ under $v$ if there exists $e_i \in \Entity$ such that $\Model, v[x_i \mapsto e_i] \models F$. Rest of the relations are defined recursively, following the structure of formulas.

The model $\Model = \langle \Entity, J \rangle$, together with a valuation $v$, determines the truth values of formulas via the satisfaction (entailment) relation $\models$, defined inductively on the structure of formulas. In particular, $\mathcal{M},v\models P(t_1,t_2)$ iff $(v(t_1), v(t_2))\in J(P)$; $\mathcal{M},v\models \neg F$ iff $\mathcal{M},v\not\models F$; $\mathcal{M},v\models F\lor G$ iff $\mathcal{M},v\models F$ or $\mathcal{M},v\models G$; and $\mathcal{M},v\models(\exists x_i)F$ iff there exists some $e_i$ such that $\mathcal{M},v[x_i \mapsto e_i]\models F$. The remaining cases are defined analogously.

If $\mathcal{M},v\models F$, we say that $F$ is true in $\mathcal{M}$ under the valuation $v$. For a fixed model $\mathcal{M}$ and formula $F$, the truth of $F$ depends only on the values assigned by $v$ to the free variables $x_i$ of $F$. In particular, if $F$ contains no free variables, then either $\mathcal{M},v\models F$ holds for all valuations $v$ or for none.

Accordingly, if $\mathcal{M},v\models F$ for all valuations $v$, we write $\mathcal{M}\models F$ and say that $F$ is true in $\mathcal{M}$, or equivalently, that $\mathcal{M}$ is a model of $F$. More generally, if $R$ is a set of formulas and $\mathcal{M}\models R$ for every $R_i\in X$, then $\mathcal{M}$ is a model of $X$, written $\mathcal{M}\models X$. Finally, if $\mathcal{M}\models F$ holds for all models $\mathcal{M}$, we write $\models F$ and call $F$ (logically) valid.

\section{Distributive Normal Forms}
Distributive normal forms for full first-order logic generalize the \emph{complete disjunctive normal forms} (CDNF), also known as canonical DNFs or sums of minterms, from propositional and monadic first-order logic. In propositional logic, a CDNF is a disjunction of conjunctions in which every variable appears exactly once in each conjunction, either positively or negated. Each conjunction corresponds to a truth assignment that makes the formula true, so two formulas are logically equivalent if and only if they have the same CDNF. For example, over propositional variables $p$ and $q$, the formula $p \lor q$ has the CDNF
$(p \land q)\ \lor\ (p \land \neg q)\ \lor\ (\neg p \land q)$. Constituents play an analogous role in first-order logic: they function as complete normal forms, providing maximally informative descriptions of possible configurations.

\section{Probability Assignment to Closed Sentences and Open Formulas with Valuations}
It is important to distinguish clearly between a \emph{model} and a \emph{valuation}. A model $\mathcal{M} = \langle \Entity,J \rangle$ specifies the domain of entities $\Entity$ and the interpretation of predicate symbols $J$, thereby fixing the relational structure of the world. A valuation $v : \Variable \to \Entity$, on the other hand, assigns elements of the domain $\Entity$ to the variables $\Variable$ of the language and serves to evaluate formulas relative to a given model.

In probability theory, probabilities are assigned to elements of a $\sigma$-algebra of events, i.e., to subsets of an outcome space $\Omega$ that are closed under complement and countable union. In logical probability, the outcome space is typically taken to be a set of possible \emph{states of the world}, represented by models $\mathcal M$ (or by model--valuation pairs $(\mathcal M,v)$). Accordingly, a closed first-order sentence $\varphi$, examples of which include $P(e_1, e_2), \, \exists x_1 P(x_1, e_2), \, \exists x_1 \forall x_2 P(x_1, x_2)$ for entities $e_1, e_2$ and variables $x_1, x_2$, has a determinate truth value in every model, in the sense that either $\mathcal M \models \varphi$ or $\mathcal M \not\models \varphi$. Thus, it naturally determines an event
\[
[\varphi] := \{\, \mathcal M \in \Omega : \mathcal M \models \varphi \,\}.
\]
A probability measure $\mu$ on an appropriate $\sigma$-algebra over $\Omega$ therefore induces a probability assignment to closed sentences via $\Pr(\varphi) := \mu([\varphi])$, with the usual identities $\Pr(\neg \varphi)=1-\Pr(\varphi)$ and $\Pr(\varphi \lor \psi)=\Pr(\varphi)+\Pr(\psi)-\Pr(\varphi \land \psi)$.  

By contrast, an open formula $\varphi(x_1,\dots,x_n)$, examples of which include $P(v[x_1 \mapsto e_1], v[x_2 \mapsto e_2])$, does not possess a truth value independently of a valuation $v : \Variable \to \Entity$ and hence does not, by itself, determine a subset of $\Omega$. Once a valuation is fixed, however, the satisfaction relation $\mathcal M,v \models \varphi$ induces the event
\[
[\varphi,v] := \{\, \mathcal M \in \Omega : \mathcal M,v \models \varphi \,\},
\]
which is semantically equivalent to the event determined by the corresponding grounded sentence $\varphi(c_1,\dots,c_n)$ with $c_i = v(x_i)$, even though the open formula and the closed sentence remain syntactically distinct. Consequently, probabilities may be coherently assigned \emph{either to closed sentences or to open formulas relative to a valuation (or parameters)}, but not to open formulas as bare syntactic objects. This distinction reflects why only truth-valued formulas can serve as events satisfying Kolmogorov’s axioms.

This point becomes especially important when considering \emph{constituents}. Constituents—understood as complete disjunctive normal forms of a first-order language—are purely syntactic constructions. They may contain free variables, provided these variables occur uniformly across all disjuncts. The presence of free variables does not obstruct their syntactic definition or manipulation, since normal-form construction proceeds independently of interpretation. However, the existence of such syntactic constituents does \emph{not} warrant the assignment of probabilities to them in the absence of a model and, when relevant, a valuation. Open formulas or constituents containing free variables do not determine sets of models and therefore fail to define measurable events; assigning probabilities to them directly would undermine the set-theoretic foundations of probability by leaving complements, unions, and intersections undefined.

In inductive-logical frameworks, prior probabilities are therefore assigned to closed constituents. These priors depend only on structural symmetries of the language and not on particular valuations of variables, in the spirit of Carnap-style inductive logic \cite{carnap_logical_1962}. Valuations enter only at the posterior stage, through Bayesian conditioning on evidence. Posterior probabilities $p(C^{w,k}\mid\Evidence)$ reflect how frequently particular models (and therefore constituents they logically entail) are supported by the observed data as is the true but unknown model $\mathcal M$ of the world $\Entity$. Since the domain of entities is assumed to be countably infinite, we further assume that only a finite subset of valuations is instantiated in the evidence.

Once a posterior distribution over constituents has been determined, probabilities can then be assigned to closed sentences, or to open formulas relative to specified valuations, by considering the set of models in which they are satisfied. In this way, free variables may appear at the syntactic level of constituent construction, but probabilistic interpretation ultimately applies only to closed sentences or to open formulas evaluated under an explicit valuation.

\end{document}